\documentclass[conference,letterpaper]{IEEEtran}

\usepackage[utf8]{inputenc} 
\usepackage[T1]{fontenc}
\usepackage{url}
\usepackage{ifthen}
\usepackage{cite}
\usepackage{enumitem}
\usepackage{balance}
\usepackage[cmex10]{amsmath} 

\usepackage{amsfonts,amsthm,amssymb}
\usepackage{courier}
\usepackage{graphicx}
\usepackage{textcomp}

\newtheorem{theo}{\textbf{Theorem}}
\newtheorem{corl}{\textbf{Corollary}}

\newtheorem{lem}{\textbf{Lemma}}

\theoremstyle{definition}
\newtheorem{remk}{Remark}

\usepackage{cite}
\usepackage{mathtools}
\usepackage{mathrsfs}
\usepackage{mathtools}

\newcommand*{\Scale}[2][4]{\scalebox{#1}{$#2$}}
\newcommand\numberthis{\addtocounter{equation}{1}\tag{\theequation}}

\begin{document}
\title{On the Communication Latency of Wireless Decentralized Learning} 


\author{%
  \IEEEauthorblockN{Navid Naderializadeh}
  \IEEEauthorblockA{Intel Labs, Santa Clara, CA \\ nn245@cornell.edu}
}


\maketitle

\begin{abstract}
We consider a wireless network comprising $n$ nodes located within a circular area of radius $R$, which are participating in a decentralized learning algorithm to optimize a global objective function using their local datasets. To enable gradient exchanges across the network, we assume each node communicates only with a set of neighboring nodes, which are within a distance $R n^{-\beta}$ of itself, where $\beta\in(0,\frac{1}{2})$. We use tools from network information theory and random geometric graph theory to show that the communication delay for a single round of exchanging gradients on all the links throughout the network scales as $\mathcal{O}\left(\frac{n^{2-3\beta}}{\beta\log n}\right)$, increasing (at different rates) with both the number of nodes and the gradient exchange threshold distance.
\end{abstract}


\section{Introduction}
With the advent of novel powerful computing platforms, alongside the availability of large-scale datasets, machine learning (ML), and particularly, deep learning, have gained significant interest in recent years~\cite{Goodfellow-et-al-2016}. Such developments have also contributed to the invention of more advanced ML architectures and more efficient training mechanisms~\cite{lecun2015deep}, which have resulted in state-of-the-art performance in many domains, such as computer vision~\cite{krizhevsky2012imagenet}, natural language processing~\cite{sutskever2014sequence}, health-care~\cite{leung2014deep}, etc.

More recently, however, there has been an increasing awareness in consumers of services, which are driven by ML models, regarding the privacy of their data. Depending on how sensitive the data type is or how often it is collected, each user has their own privacy concerns and preferences~\cite{soups}. Such trends have been coincident with the proliferation of mobile computing solutions, which provide devices, such as smart-home devices, cell phones, laptops, and drones, with strong computation capabilities~\cite{poushter2016smartphone}.

These societal and technical trends have given rise to paradigms such as federated and decentralized learning, where the generated data by each device stays on-board to protect its privacy~\cite{mcmahan2016communication,kamp2018efficient}. To compensate for that, (part of) the computation is also shifted to be done locally at the end-user devices. It has been shown that in many cases, distributing the learning process over different nodes incurs negligible performance loss compared to centralized training approaches~\cite{zhang2013communication}.

However, one major bottleneck in all the aforementioned paradigms is the communication network between the learning nodes. As the data points generated by each node differ from the rest of the rest of the network, the nodes need to periodically communicate with each other so that they all converge to the same model, rather than diverging to completely different models. If the communication that needs to occur between the nodes in the network induces sizeable delays, it can significantly lengthen the convergence time across the network, as it can totally dominate the computation delay at the learning nodes.

This phenomenon has motivated a massive body of recent work on dealing with the communication delays for federated and decentralized learning. In~\cite{chen2018lag}, a setting with a single server and multiple worker nodes is considered, where at each iteration, a subset of worker nodes are selected, either by the server or by the worker nodes themselves, to send their gradients to the server. In~\cite{scaman2018optimal}, a simple network of multiple worker nodes is considered, over which they can all exchange their computation results with a fixed amount of delay, in conjunction with a server which aggregates all the results and sends back updated parameters to the worker nodes. In~\cite{koloskova2019decentralized}, gossiping algorithms and convergence guarantees are provided for decentralized optimization with compressed communication. In~\cite{wang2019matcha}, it is shown how specific connectivity of the communication network topology among learning nodes affects the speed of convergence. In~\cite{basu2019qsparse}, convergence results are derived for a combination of quantization, sparsification and local computation in a distributed computation setting with a single master and multiple worker nodes. In~\cite{reisizadeh2019robust}, a deadline-based approach for minibatch gradient computing at each computing node is proposed, such that the minibatch size is adaptive to the computation capabilities of each node, hence making the scheme robust to stragglers.

Most of the above works deal with an abstract model for the communication network among the learning nodes. One particularly interesting communication paradigm to consider is wireless communication, especially as operators around the world roll out their 5G network infrastructure. There have been some recent works that have considered wireless constraints, mostly in the context of federated learning~\cite{amiri2019machine,ahn2019wireless,zeng2019energy,yang2019scheduling,amiri2020update}.

In this paper, we consider the decentralized learning scenario over a network of learning nodes connected together through a shared wireless medium. Considering the nature of wireless networks, in which nodes in proximity can more efficiently communicate with each other, while interfering at concurrent transmissions, we attempt to characterize the communication delay for exchanging the gradients among the learning nodes over the wireless network topology. In particular, we consider a setting similar to~\cite{wang2019matcha}, where at each time, a set of non-interfering gradient exchanges are scheduled to happen simultaneously. Using the results on the optimality of treating interference as noise in interference networks~\cite{geng2015optimality}, we present an algorithm for gradient exchanges in wireless decentralized learning akin to the information-theoretic link scheduling that was proposed in~\cite{naderializadeh2014itlinq} for the case of device-to-device networks.

We utilize tools from random geometric graph theory to characterize the asymptotic communication latency for exchanging gradients in the aforementioned decentralized setting framework. In particular, we consider a network of $n$ learning nodes located within a circle of radius $R$, where each node exchanges gradients with its neighboring nodes, which are within a distance $R n^{-\beta}$ of itself, where $\beta\in(0,\frac{1}{2})$ is a variable that controls the density of the gradient exchange topology. This threshold distance needs to decrease with $n$, as the entire network needs to remain connected to guarantee the convergence of the decentralized learning algorithm. We show that as $n\rightarrow\infty$, the communication latency scales as $\mathcal{O}\left(\frac{n^{2-3\beta}}{\beta\log n}\right)$, increasing with the number of users, and decreasing with $\beta$. This result provides insights on how much communication time is needed in a wireless decentralized learning scenario, where more gradient exchanges leads to longer communication latencies, but faster convergence rates.


\section{System Model}\label{sec:model}
Consider a wireless network consisting of $n$ nodes $[n] \triangleq \{1,2,...,n\}$ dropped uniformly at random within a circular area of radius $R$.
Assume that each node $i \in [n]$ has access to a set of data points $\mathcal{D}_i$, and the goal is to minimize a global loss function $f$, defined over a set of optimization parameters $\boldsymbol{w}\in\mathbb{R}^d$, using the overall dataset across the network as
\begin{align*}
\min_{\boldsymbol{w}\in\mathbb{R}^d} f(\boldsymbol{w})
&= \min_{\boldsymbol{w}\in\mathbb{R}^d}\frac{1}{n} \sum_{i=1}^n  f_i(\boldsymbol{w})\\
&= \min_{\boldsymbol{w}\in\mathbb{R}^d}\frac{1}{n} \sum_{i=1}^n
\mathbb{E}_{\mathbf{x} \sim \mathcal{D}_i}\left[l(\boldsymbol{w};\mathbf{x})\right],
\end{align*}
where $f_i(\boldsymbol{w})$ is the local loss function at node $i\in[n]$, and $l(\boldsymbol{w};\mathbf{x})$ is the stochastic loss function for sample $\mathbf{x}$ given model parameters $\boldsymbol{w}$. In order to solve this problem, decentralized stochastic gradient descent (SGD) can be utilized to minimize the objective function in an iterative fashion. In decentralized SGD, the system is run over multiple iterations, where at each iteration, each node performs a local computation of the gradient of the objective function with respect to the set of optimization parameters $\boldsymbol{w}$ over (a minibatch of) its local dataset, following which the gradients are exchanged among nodes prior to the beginning of the next iteration.

Due to the path-loss and fading effects in wireless communications, nodes can more easily communicate to their closer neighbors than farther ones. Therefore, we define the \emph{communication graph} as the network topology which dictates how nodes exchange gradients with their neighboring nodes, and we model it as an undirected random geometric graph (RGG) $G_{\mathsf{comm}}=(\mathcal{V}_{\mathsf{comm}}, \mathcal{E}_{\mathsf{comm}})$, where $\mathcal{V}_{\mathsf{comm}}=[n]$ is the set of all nodes in the network, and for every $i,j\in \mathcal{V}_{\mathsf{comm}}$, where $i \neq j$, $(i,j)\in \mathcal{E}_{\mathsf{comm}}$ if and only if 
$D_{ij} \leq D_{\mathsf{comm}}$, where $D_{ij}$ denotes the distance between nodes $i$ and $j$, and $D_{\mathsf{comm}}$ is the \emph{threshold distance} for gradient exchange; i.e., two nodes can exchange their gradients with each other if and only if they are located within a distance of at most $D_{\mathsf{comm}}$.

However, activating multiple gradient exchanges over the wireless channel at the same time will lead to interference, which can significantly reduce the network performance in terms of the throughput, and therefore, the communication delay. To capture the interference among concurrent wireless transmissions, we also define a \emph{conflict graph} $G_{\mathsf{conf}}=(\mathcal{V}_{\mathsf{conf}}, \mathcal{E}_{\mathsf{conf}})$. In this graph, each vertex represents a communication link in the original communication graph, i.e., $\mathcal{V}_{\mathsf{conf}}=\mathcal{E}_{\mathsf{comm}}$. Moreover, there is an edge between two vertices in $\mathcal{V}_{\mathsf{conf}}$ if their activations are in conflict; i.e., if transmitting data (i.e., gradients) on those links at the same time \emph{strongly interfere} on each other. Since the level of interference also depends on the distance of transmitting/receiving nodes, we introduce a \emph{conflict distance} $D_{\mathsf{conf}}$, where for two vertices $(i_1,j_1), (i_2,j_2) \in \mathcal{V}_{\mathsf{conf}}$, there is an edge between $(i_1,j_1)$ and $(i_2,j_2)$, i.e., $((i_1,j_1), (i_2,j_2)) \in \mathcal{E}_{\mathsf{conf}}$, if and only if
\begin{align*}
\min\{D_{i_1,i_2}, D_{i_1,j_2}, D_{j_1,i_2}, D_{j_1,j_2} \} \leq D_{\mathsf{conf}},
\end{align*}
which implies that at least one node in $(i_1,j_1)$ is within conflict distance of $(i_2,j_2)$. Note that for the case of $i_1=i_2=i$, $D_{i_1,i_2}=0$, implying that there is a conflict between $(i,j_1)$ and $(i,j_2)$, for any two neighbors $j_1,j_2$ of node $i$ in the original communication graph. This means that a node cannot communicate with two nodes at the same time (i.e., half-duplex and single frequency band constraints).

Given the above definitions, our goal is to determine the asymptotic behavior of the normalized gradient exchange latency $\delta$ (as $n\rightarrow\infty$), which is defined as the delay for completing the exchange of 1 bit of gradients on all the links of the communication graph. Assuming that the communication delay in the network dominates the gradient computation delay at each node, the normalized gradient exchange latency $\delta$ characterizes the wall-clock run time per iteration for decentralized SGD on a wireless communication network of learning nodes.


\subsection{Wireless Communication Model}
We assume each node is equipped with a single transmit/receive antenna, and all transmissions happen in a synchronous time-slotted manner on a single frequency band. We restrict the transmission strategies to an on/off pattern: At each time slot, a node either transmits a message to another node with full power $P$ or stays completely silent. We use $\mu_i(t)\in\{0,1\}$ as a transmission status indicator of node $i$ at time slot $t$; i.e., $\mu_i(t)=1$ if and only if node $i$ is transmitting with full power at time slot $t$. On the receiver side, we adopt the simple and practical scheme of treating interference as noise (TIN), where each node decodes its desired message, while treating the interference from all other concurrent transmissions as noise. Letting $N$ denote the noise variance, the rate achieved on a link from node $i$ to node $j$ at time $t$ can be written as
\begin{align}\label{eq:rate1}
R_{ij}(t) = \frac{\mu_i(t) \cdot P \cdot G_{ij}}{\sum_{k\in [n]\setminus\{i,j\}} \mu_k(t) \cdot P \cdot G_{kj} + N},
\end{align}
where $G_{ij}$ denotes the channel gain on the link between nodes $i$ and $j$. In this paper, we adopt a single-slope path-loss model for the channel gains, where the channel gain at distance $D$ can be written as
\begin{align*}
G(D) = G_0 D^{-\alpha},
\end{align*}
where $G_0$ is the reference channel gain at a distance of $1m$, and $\alpha \geq 2$ denotes the path-loss exponent. This implies that the achievable rate in \eqref{eq:rate1} can be written as
\begin{align*}
R_{ij}(t) &= \frac{\mu_i(t) \cdot P \cdot G_0 \cdot D_{ij}^{-\alpha}} {\sum_{k\in [n]\setminus\{i,j\}} \mu_k(t) \cdot P \cdot G_0 \cdot D_{kj}^{-\alpha} + N}\\
&= \frac{\mu_i(t)\gamma D_{ij}^{-\alpha}} {\sum_{k\in [n]\setminus\{i,j\}} \mu_k(t) \gamma D_{kj}^{-\alpha} + 1},
\end{align*}
where $\gamma\triangleq \frac{P \cdot G_0}{N}$ denotes the signal-to-noise ratio (SNR) at a distance of $1m$.

\section{Forming the Communication and Conflict Graphs}
The communication network topology needs to be carefully designed, as decentralized SGD will not converge if the gradient exchange communication graph is disconnected~\cite{wang2019matcha}. We resort to the following lemma, which provides a sufficient condition for connectivity of random geometric graphs.
\begin{lem}[Corollary 3.1 in~\cite{gupta1999critical}]\label{lem:RGGconnect}
In an RGG with $n$ nodes and a threshold distance of $r(n)$, the graph is connected with probability one (as $n\rightarrow\infty$) if $\pi r^2(n)=\frac{\log(n) + c(n)}{n}$, where $c(n)\overset{n\uparrow}{\longrightarrow}\infty$.\footnote{In this paper, we use the short-hand notation $\log(\cdot)$ to denote the natural logarithm operation $\log_e(\cdot)$.}
\end{lem}

In light of Lemma~\ref{lem:RGGconnect}, for the communication graph, we set the gradient exchange threshold distance as
\begin{align}\label{eq:def_Dcomm}
D_{\mathsf{comm}} = R n^{-\beta}, ~\beta\in(0,\tfrac{1}{2}),
\end{align}
which decreases as the number of nodes increases so as to satisfy the condition in Lemma~\ref{lem:RGGconnect}, hence maintaining the connectivity of the entire graph.

Now, to build the conflict graph, we use the following result, derived in~\cite{geng2015optimality}, for approximate information-theoretic optimality of TIN in wireless networks.

\begin{theo}[Theorem 4 in~\cite{geng2015optimality}]\label{thm:TIN}
Consider a wireless network with $K$ transmitter-receiver pairs $\left\{\left( \mathsf{Tx}_i,\mathsf{Rx}_i \right)\right\}_{i=1}^K$, where $\mathsf{SNR}_i$ denotes the signal-to-noise ratio between $\mathsf{Tx}_i$ and $\mathsf{Rx}_i$, and $\mathsf{INR}_{ij}$ denotes the interference-to-noise ratio between $\mathsf{Tx}_i$ and $\mathsf{Rx}_j$. Then, under the following condition,
\begin{align*}
\mathsf{SNR}_i \geq \mathsf{INR}_{ij} \cdot \mathsf{INR}_{li}, \qquad \forall i\in[K], \forall j,l\in[K]\setminus\{i\},
\end{align*}
TIN achieves the entire information-theoretic capacity region of the network (as defined in \cite{geng2015optimality}) to within a gap of $\log_2 3K$ per dimension.
\end{theo}
Theorem \ref{thm:TIN} immediately leads to the following corollary.
\begin{corl}\label{corl:TIN}
In a network with $K$ transmitter-receiver pairs, if the minimum SNR and the maximum INR across the whole network (denoted by $\mathsf{SNR}_{\min}$ and $\mathsf{INR}_{\max}$, respectively) satisfy $\mathsf{INR}_{\max} \leq \sqrt{\mathsf{SNR}_{\min}}$, then TIN is information-theoretically optimal to within a gap of $\log_2 3K$ per dimension.
\end{corl}

As mentioned in Section~\ref{sec:model}, the received power at distance $D$ can be written as $P G_0 D^{-\alpha}$. Hence, given the RGG nature of the communication and conflict graphs, we can bound the SNR and INR values across the network as
\begin{align}
\mathsf{SNR}_{\min} &\geq \frac{P G_0 D_{\mathsf{comm}}^{-\alpha}}{N} = \gamma D_{\mathsf{comm}}^{-\alpha},\label{eq:SNRbound}\\
\mathsf{INR}_{\max} &\leq \frac{P G_0 D_{\mathsf{conf}}^{-\alpha}}{N} = \gamma D_{\mathsf{conf}}^{-\alpha}.\label{eq:INRbound}
\end{align}
Therefore, \eqref{eq:SNRbound}-\eqref{eq:INRbound} together with Corollary~\ref{corl:TIN} imply that a sufficient condition for the optimality of TIN for exchanging the gradients is
\begin{align*}
\gamma D_{\mathsf{conf}}^{-\alpha} &\leq \sqrt{\gamma D_{\mathsf{comm}}^{-\alpha}}\\
\Leftrightarrow D_{\mathsf{conf}} &\geq \gamma^{\frac{1}{2\alpha}} \sqrt{D_{\mathsf{comm}}}.
\end{align*}
Thus, to guarantee the optimality of TIN, while having the sparsest conflict graph, we set the conflict distance as
\begin{align}\label{eq:def_Dconf}
D_{\mathsf{conf}} = \gamma^{\frac{1}{2\alpha}} \sqrt{D_{\mathsf{comm}}} = \gamma^{\frac{1}{2\alpha}} \sqrt{R} n^{-\beta/2}.
\end{align}

\section{Main Result}

In this section, we present our main result on the time needed for exchanging gradients over the communication graph as follows.
\begin{theo}\label{thm:main}
For a sufficiently large network of learning nodes ($n \rightarrow \infty$), the normalized gradient exchange latency satisfies
\begin{align}\label{eq:main_thm}
\delta &< \frac{1 + \frac{2\gamma^{\frac{1}{\alpha}}}{R} n^{2-3\beta}}{\log_2\left( 1 + 2\sqrt{\gamma}R^{-\frac{\alpha}{2}} n^{\frac{\alpha\beta}{2} + 2\beta - 2} \left(1 + \frac{2\gamma^{\frac{1}{\alpha}}}{R} n^{2-3\beta}\right) \right)}.
\end{align}
\end{theo}

\begin{remk}
Theorem~\ref{thm:main} implies that the normalized gradient exchange latency can be upper-bounded in an order-wise fashion (for $n\rightarrow\infty$) as
\begin{align}\label{eq:thm_orderwise}
\delta < \mathcal{O}\left(\frac{n^{2-3\beta}}{\beta\log n}\right).
\end{align}
\end{remk}

Theorem~\ref{thm:main} characterizes an achievable normalized gradient exchange latency over the communication graph. Figure~\ref{fig:thm} demonstrates how this latency changes with $n$ and $\beta$ for the case where nodes are dropped within a circular area of radius $100$m, transmit power is assumed to be $30$dBm, noise power spectral density is taken to be $-174$dBm/Hz, the bandwidth is $10$MHz, the path-loss exponent is equal to $2$, and the reference channel gain is set to $G_0 = 10^{-7}$. As demonstrated by~\eqref{eq:main_thm} and its order-wise approximation in~\eqref{eq:thm_orderwise}, as well Figure~\ref{fig:thm}, the delay of exchanging gradients over all links in the conflict graph monotonically increases with $n$, which is expected as increasing the network size, while keeping the communication graph connected, will require an increasing number of gradient exchanges among neighboring nodes.

\begin{figure}[t]
\center
\includegraphics[trim=1.5in .6in .23in 1.05in, clip, width=0.485\textwidth]{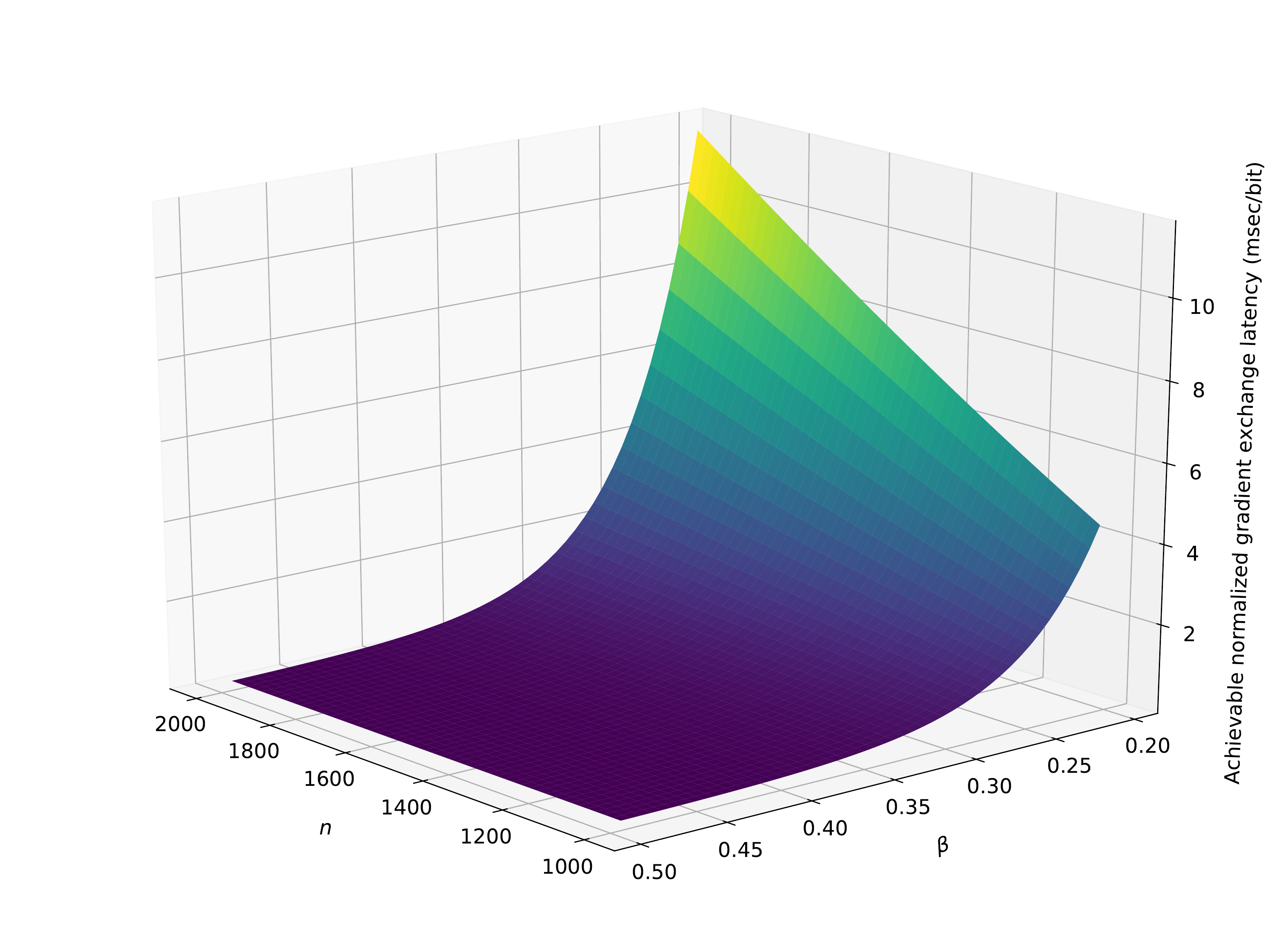}
\caption{Variations of the achievable normalized gradient exchange latency for a network with $1000-2000$ nodes dropped in a circular area of radius $100$m.}
\label{fig:thm}
\end{figure}

On the other hand, the latency decreases (approximately) exponentially with $\beta$. As per~\eqref{eq:def_Dcomm}, $\beta$ determines the threshold distance for gradient exchange among adjacent nodes; Increasing $\beta$ will reduce the number of neighbors with which each node exchanges gradients, and this provides a significant saving in terms of the communication latency. Note that this comes at the expense of slower convergence rate for the global loss function, as it will take longer for each node to obtain access to the gradients from datasets available in farther nodes.

\section{Achievable Scheme}

In this section, we prove our main result in Theorem~\ref{thm:main} by providing an achievable scheme for gradient exchange on all links in the communication graph and characterizing an upper bound on its achievable normalized gradient exchange latency.

Given the communication and conflict graphs, the nodes can exchange gradients with their neighbors in the communication graph as long as their exchanges are non-conflicting; i.e., there is not an edge between them in the conflict graph. This leads to the notion of independent sets on the conflict graph, where each such independent set contains a set of nodes such that there is no edge between them. This is closely related to the notion of information-theoretic independent sets as defined in~\cite{naderializadeh2014itlinq} for device-to-device communication networks. It is also analog to the concept of matchings on the communication topology as considered in~\cite{wang2019matcha}, where now the interference between active communication links is also taken into account.

We first start with the following lemma, in which we characterize a lower bound on the symmetric rate within an independent set of the conflict graph, defined as the rate that can be simultaneously achieved by all the corresponding active links in the communication graph.

\begin{lem}\label{lem:rsym}
For any independent set $\mathcal{S} \subseteq \mathcal{V}_{\mathsf{conf}}$ in $G_{\mathsf{conf}}$, the symmetric rate is lower-bounded by
\begin{align}
R_{\mathsf{sym}, \mathcal{S}} > \log_2\left( 1 + \frac{\sqrt{\gamma D_{\mathsf{comm}}^{-\alpha}}}{|\mathcal{S}|}\right).
\end{align}
\end{lem}

\begin{proof}
For every vertex $(i,j)\in\mathcal{S}$, the achievable rate on the corresponding link from node $i$ to node $j$ in $G_{\mathsf{comm}}$ can be written as
\begin{align*}
R_{ij} &= \log_2(1+\mathsf{SINR}_{ij}) \\
& = \log_2\left( 1 + \frac{\mathsf{SNR}_{ij}}{1 + \sum_{l\neq i : \exists m \text{ s.t. } (l,m)\in\mathcal{S}} \mathsf{INR}_{lj}} \right)\\
& = \log_2\left( 1 + \frac{\gamma D_{ij}^{-\alpha}}{1 + \sum_{l\neq i : \exists m \text{ s.t. } (l,m)\in\mathcal{S}} \gamma D_{lj}^{-\alpha}} \right)\\
& > \log_2\left( 1 + \frac{\gamma D_{\mathsf{comm}}^{-\alpha}}{1 + \sum_{l\neq i : \exists m \text{ s.t. } (l,m)\in\mathcal{S}} \gamma D_{\mathsf{conf}}^{-\alpha}} \right)\numberthis\label{eq:distbound}\\
& = \log_2\left( 1 + \frac{\gamma D_{\mathsf{comm}}^{-\alpha}}{1 + (|\mathcal{S}| - 1) \gamma D_{\mathsf{conf}}^{-\alpha}} \right)\\
&\overset{n\uparrow}{>} \log_2\left( 1 + \frac{\sqrt{\gamma D_{\mathsf{comm}}^{-\alpha}}}{|\mathcal{S}|}\right),\numberthis\label{eq:tin_noise}
\end{align*}
where \eqref{eq:distbound} follows from the fact that link $(i,j)$ is present in the communication graph, hence their distance satisfies $D_{ij} \leq D_{\mathsf{comm}}$, while the link between nodes $(i,j)$ and $(l,m)$ is not present in the conflict graph, implying that $D_{lj} > D_{\mathsf{conf}}$. Moreover, \eqref{eq:tin_noise} follows from the definition of $D_{\mathsf{conf}}$ in~\eqref{eq:def_Dconf}, and from the fact that as $n\rightarrow\infty$, the interference grows larger than noise, i.e., $D_{\mathsf{conf}}^{-\alpha}\gg1$. As all nodes $(i,j) \in \mathcal{S}$ are able to achieve this communication rate, the proof is complete.
\end{proof}

Next, we present the following lemma, which provides an upper bound on the chromatic number of the conflict graph.
\begin{lem}\label{lem:bound_chrom}
The chromatic number $\chi_{\mathsf{conf}}$ of the conflict graph can be asymptotically upper-bounded by
\begin{align*}
\chi_{\mathsf{conf}} \leq 1 + \frac{2 \gamma^{\frac{1}{\alpha}}}{R} n^{2-3\beta}.
\end{align*}
\end{lem}

\begin{proof}
Considering each vertex $(i,j)$ in the conflict graph, its degree can be upper-bounded as
\begin{align}
deg_{(i,j)} &\leq |\{(k,l) \in \mathcal{V}_{\mathsf{conf}} : D_{ik} \leq D_{\mathsf{conf}}\}| \nonumber \\
&\qquad + |\{(k,l) \in \mathcal{V}_{\mathsf{conf}} : D_{jk} \leq D_{\mathsf{conf}}\}|\nonumber\\
&=|\{l: D_{kl} \leq D_{\mathsf{comm}}\}| \cdot \nonumber \\
&\qquad\Big(|\{k: D_{ik} \leq D_{\mathsf{conf}}\}| \nonumber\\
&\qquad + |\{k: D_{jk} \leq D_{\mathsf{conf}}\}|\Big)\nonumber\\
&\leq 2 \Delta_{RGG(n,D_\mathsf{conf})} \Delta_{\mathsf{comm}},\label{eq:prod_maxdegrees}
\end{align}
where $\Delta_{RGG(n,D_\mathsf{conf})}$ is the maximum degree of a random geometric graph with $n$ nodes and threshold distance of $D_\mathsf{conf}$, and $\Delta_{\mathsf{comm}}$ is the maximum degree of $G_{\mathsf{comm}}$, which is a random geometric graph with $n$ nodes and threshold distance of $D_\mathsf{comm}$. As per equation (4) in~\cite{decreusefond:hal-00864303}, \eqref{eq:prod_maxdegrees} can be upper bounded by
\begin{align}\label{eq:prod_cliques}
deg_{(i,j)} &\leq 2 \cdot \omega_{RGG(n,2D_\mathsf{conf})} \cdot \omega_{RGG(n,2D_\mathsf{comm})},
\end{align}
where $\omega_{RGG(n,r)}$ denotes the clique number of a random geometric graph with $n$ nodes and threshold distance $r$, defined as the size of the largest clique in the graph, i.e., the maximal subset of vertices in which every two vertices are connected. 

Now, we can leverage the bounds in the following theorem from~\cite{mcdiarmid2011chromatic} on the clique number of random geometric graphs to upper bound~\eqref{eq:prod_cliques}.
\begin{theo}[Theorem 1.2 in~\cite{mcdiarmid2011chromatic}]\label{thm:rgg_omega}
For a $d$-dimensional random geometric graph with $n$ nodes and threshold distance $r \overset{n\uparrow}{\longrightarrow} 0$, if $\frac{\ln n}{nr^d}\overset{n\uparrow}{\longrightarrow} 0$, then its clique number, denoted by $\omega_{RGG(n,r)}$, satisfies
\begin{align*}
\frac{\omega_{RGG(n,r)}}{\frac{\emph{vol}(B)}{2^d}\sigma nr^d} \overset{n\uparrow}{\longrightarrow} 1,
\end{align*}
where $B$ is the unit ball in $\mathbb{R}^d$ and $\sigma$ is the maximum density of the distribution of nodes in $\mathbb{R}^d$. For Euclidean distance in $\mathbb{R}^2$ and uniform distribution of nodes within a circle of radius $R$, $\text{vol}(B) = \pi$ and $\sigma = \frac{1}{\pi R^2}$.
\end{theo}

For the graph $RGG(n,2D_\mathsf{conf})$, we have $n(2D_\mathsf{conf})^2 \overset{\eqref{eq:def_Dconf}}{=} 4n(\gamma^{\frac{1}{2\alpha}} \sqrt{R} n^{-\beta/2})^2 = 4R \gamma^{\frac{1}{\alpha}} n^{1-\beta}$. Given the fact that $1-\beta\in(\frac{1}{2},1)$, we can invoke Theorem~\ref{thm:rgg_omega} to (almost-surely) continue~\eqref{eq:prod_cliques} as
\begin{align}
deg_{(i,j)} &\leq 2
\cdot \left(\frac{n (2D_\mathsf{conf})^2}{4R^2}\right)
\cdot \omega_{RGG(n,2D_\mathsf{comm})}\nonumber \\
&= \frac{2 \gamma^{\frac{1}{\alpha}}}{R} n^{1-\beta} \cdot \omega_{RGG(n,2D_\mathsf{comm})}.\label{eq:deg_bound_omega_dcomm_left}
\end{align}
Furthermore, for the graph $RGG(n,2D_\mathsf{comm})$, we have $n(2D_\mathsf{comm})^2 \overset{\eqref{eq:def_Dcomm}}{=} 4n(R n^{-\beta})^2=4R^2 n^{1-2\beta}$, and since $1-2\beta\in(0,1)$, we can again use Theorem~\ref{thm:rgg_omega} to continue~\eqref{eq:deg_bound_omega_dcomm_left} as
\begin{align}
deg_{(i,j)} &\leq \frac{2 \gamma^{\frac{1}{\alpha}}}{R} n^{1-\beta} \cdot \left(\frac{n (2D_\mathsf{comm})^2}{4R^2}\right)\nonumber\\
&=\frac{2 \gamma^{\frac{1}{\alpha}}}{R} n^{2-3\beta}.\label{eq:deg_bound_final1}
\end{align}

Using a greedy coloring algorithm on the conflict graph, its chromatic number can be upper bounded by $1 + \Delta_{\mathsf{conf}}$, where $\Delta_{\mathsf{conf}}$ is the maximum degree of the vertices in $G_{\mathsf{conf}}$. Combined with~\eqref{eq:deg_bound_final1}, this completes the proof.
\end{proof}

Having Lemmas~\ref{lem:rsym} and~\ref{lem:bound_chrom}, we now proceed to prove Theorem~\ref{thm:main}. Suppose that we have a proper coloring on the conflict graph with $\chi_{\mathsf{conf}}$ colors, where the independent set corresponding to each color $k\in\{1,...,\chi_{\mathsf{conf}}\}$ is denoted by $\mathcal{S}_k$. Then, assuming that all independent sets use time-sharing to exchange the gradients, we can bound the normalized gradient exchange latency as
\begin{align}\label{eq:delay1}
\delta = \sum_{k=1}^{\chi_{\mathsf{conf}}} \frac{1}{R_{\mathsf{sym}, \mathcal{S}_k}}.
\end{align}
Now, we can leverage Lemma~\ref{lem:rsym} to upper bound \eqref{eq:delay1} as
\begin{align*}
\delta &< \sum_{k=1}^{\chi_{\mathsf{conf}}} \frac{1}{\log_2\left( 1 + \frac{\sqrt{\gamma D_{\mathsf{comm}}^{-\alpha}}}{|\mathcal{S}_k|}\right)} \\
&= \chi_{\mathsf{conf}} \sum_{k=1}^{\chi_{\mathsf{conf}}} \frac{1}{\chi_{\mathsf{conf}}} \frac{1}{\log_2\left( 1 + \frac{\sqrt{\gamma D_{\mathsf{comm}}^{-\alpha}}}{|\mathcal{S}_k|}\right)}\\
&= \chi_{\mathsf{conf}} \sum_{k=1}^{\chi_{\mathsf{conf}}} \frac{1}{\chi_{\mathsf{conf}}} g(|\mathcal{S}_k|),\numberthis\label{eq:delta_f_lb}
\end{align*}
where $g(\cdot)$ is defined as
\begin{align}\label{eq:f_def}
g(x) := \frac{1}{\log_2\left( 1 + \frac{\sqrt{\gamma D_{\mathsf{comm}}^{-\alpha}}}{x}\right)}.
\end{align} 
It can be shown that $g(x)$ is concave in $x$ for $x>0$ (see Appendix~\ref{appx:concave}). Therefore, using Jensen's inequality, we can upper-bound \eqref{eq:delta_f_lb} as
\begin{align}
\delta &< \chi_{\mathsf{conf}}  f\left(\frac{1}{\chi_{\mathsf{conf}}} \sum_{k=1}^{\chi_{\mathsf{conf}}} |\mathcal{S}_k|\right)\nonumber\\
&= \frac{\chi_{\mathsf{conf}}}{\log_2\left( 1 + \frac{\sqrt{\gamma D_{\mathsf{comm}}^{-\alpha}}}{\frac{1}{\chi_{\mathsf{conf}}}\sum_{k=1}^{\chi_{\mathsf{conf}}}|\mathcal{S}_k| } \right)}\label{eq:delay_jensen}
\end{align}
Now, note that $\sum_{k=1}^{\chi_{\mathsf{conf}}} |\mathcal{S}_k|$ is equal to the total number of vertices in the conflict graph, or the edges in the communication graph; i.e.,
\begin{align*}
\sum_{k=1}^{\chi_{\mathsf{conf}}} |\mathcal{S}_k| = |\mathcal{V}_{\mathsf{conf}}| = |\mathcal{E}_{\mathsf{comm}}|.
\end{align*}
Proposition A.1 in \cite{muller2008two} suggests that the average degree of a 2-dimensional random geometric graph with $n$ nodes dropped uniformly at random within a circular area of radius $R$ and a threshold distance $r$ asymptotically converges to $n\left(\frac{r}{R}\right)^2$. Therefore, we have
\begin{align*}
\sum_{k=1}^{\chi_{\mathsf{conf}}} |\mathcal{S}_k| = |\mathcal{E}_{\mathsf{comm}}| \overset{n\uparrow}{\longrightarrow} \frac{1}{2} \cdot n\cdot n \left(\frac{D_{\mathsf{comm}}}{R}\right)^2,
\end{align*}
which together with \eqref{eq:delay_jensen} leads to
\begin{align*}
\delta &< \frac{\chi_{\mathsf{conf}}}{\log_2\left( 1 + \chi_{\mathsf{conf}} \frac{\sqrt{\gamma D_{\mathsf{comm}}^{-\alpha}}}{ \frac{n^2 D_{\mathsf{comm}}^2}{2R^2}} \right)}\\
&= \frac{\chi_{\mathsf{conf}}}{\log_2\left( 1 + 2\sqrt{\gamma}R^2\chi_{\mathsf{conf}} \frac{ D_{\mathsf{comm}}^{-\frac{\alpha}{2} - 2} }{n^2} \right)} \\
&\overset{\eqref{eq:def_Dcomm}}{=} \frac{\chi_{\mathsf{conf}}}{\log_2\left( 1 + 2\sqrt{\gamma}R^{-\frac{\alpha}{2}} n^{\frac{\alpha\beta}{2} + 2\beta - 2} \chi_{\mathsf{conf}} \right)}\numberthis\label{eq:mon_bound_delta_chi}
\end{align*}

It is not hard to verify that the bound in~\eqref{eq:mon_bound_delta_chi} is a monotonically increasing function of $\chi_{\mathsf{conf}}$ (see Appendix~\ref{appx:monotone}). Therefore, we can invoke Lemma~\ref{lem:bound_chrom} to upper bound~\eqref{eq:mon_bound_delta_chi} as in~\eqref{eq:main_thm}, hence completing the proof of Theorem~\ref{thm:main}.

\appendices

\section{Proof of Concavity of $g(x)$ in \eqref{eq:f_def} for $x>0$}\label{appx:concave}

Letting $M=\sqrt{\gamma D_{\mathsf{comm}}^{-\alpha}}$, we can write the first derivative of $g$ as
\begin{align*}
\frac{\partial g}{\partial x} &= -\frac{\frac{\partial \log_2\left( 1 + \frac{M}{x}\right)}{\partial x}}{\left[\log_2\left( 1 + \frac{M}{x}\right)\right]^2} \\
&= \frac{M\log 2}{x^2 \left( 1 + \frac{M}{x}\right)\left[\log\left( 1 + \frac{M}{x}\right)\right]^2},
\end{align*}
which leads to the second derivative of $g$ as
\begin{align}\label{eq:fsec_1}
\frac{\partial^2 g}{\partial x^2} &= -\frac{(M\log 2) \frac{\partial \left[ x^2 \left( 1 + \frac{M}{x}\right)\left[\log\left( 1 + \frac{M}{x}\right)\right]^2\right]}{\partial x}}{x^4 \left( 1 + \frac{M}{x}\right)^2\left[\log\left( 1 + \frac{M}{x}\right)\right]^4}.
\end{align}

We can write the derivative in the numerator of \eqref{eq:fsec_1} as
\begin{align*}
&\frac{\partial \left[ x^2 \left( 1 + \frac{M}{x}\right)\left[\log\left( 1 + \frac{M}{x}\right)\right]^2\right]}{\partial x} \\
&\quad= \frac{\partial \left( x^2 + Mx\right)}{\partial x}\left[\log\left( 1 + \frac{M}{x}\right)\right]^2 \\
&\qquad+ \frac{\partial \left[\log\left( 1 + \frac{M}{x}\right)\right]^2}{\partial x} \left( x^2 + Mx\right) \\
&\quad = (2x+M) \left[\log\left( 1 + \frac{M}{x}\right)\right]^2 \\
&\qquad- 2M \log\left( 1 + \frac{M}{x}\right)\\
&\ =\log\left( 1 + \frac{M}{x}\right)\left[(2x+M) \log\left( 1 + \frac{M}{x}\right)-2M\right].\numberthis\label{eq:2ndder_g_interm}
\end{align*}

Now, consider the function
\begin{align}\label{eq:h_def}
h(y) := \log(1+y) - \frac{2y}{2+y}.
\end{align}
It is easy to show that this function is non-negative for $y\geq0$. This is because $h(0)=0$, and
\begin{align*}
\frac{dh}{dy} &= \frac{1}{1+y} - \frac{2(2+y)-2y}{(2+y)^2}\\
&=\frac{(2+y)^2 - 4(1+y)}{(1+y)(2+y)^2}\\
&=\frac{y^2}{(1+y)(2+y)^2} \geq 0.
\end{align*}
Plugging in $y=\frac{M}{x}$, we can rewrite~\eqref{eq:2ndder_g_interm} as
\begin{align*}
&\frac{\partial \left[ x^2 \left( 1 + \frac{M}{x}\right)\left[\log\left( 1 + \frac{M}{x}\right)\right]^2\right]}{\partial x} \\
&\quad= \log\left( 1 + \frac{M}{x}\right)\cdot(2x+M) \cdot h\left(\frac{M}{x}\right) \geq 0,
\end{align*}
which, together with the fact that the rest of the terms in~\eqref{eq:fsec_1} are negative, completes the proof.

\section{Proof of Monotonicity of the Bound in~\eqref{eq:mon_bound_delta_chi}}\label{appx:monotone}
Let us rewrite the bound in~\eqref{eq:mon_bound_delta_chi} as $s(\chi_{\mathsf{conf}})$, where $s(\cdot)$ is defined as
\begin{align}
s(x) := \frac{x}{\log_2\left( 1 + Cx \right)},
\end{align}
with $C = 2\sqrt{\gamma}R^{-\frac{\alpha}{2}} n^{\frac{\alpha\beta}{2} + 2\beta - 2}$. We can then write the first derivative of $s$ as
\begin{align*}
\frac{\partial s}{\partial x} &= \frac{\log 2}{\left[\log\left( 1 + Cx \right)\right]^2} \left[\log(1+Cx) - \frac{Cx}{1+Cx}\right]\\
&\overset{\eqref{eq:h_def}}{=} \frac{\log 2}{\left[\log\left( 1 + Cx \right)\right]^2} \left[h(Cx) + \frac{2Cx}{2+Cx} - \frac{Cx}{1+Cx}\right] \\
&=\frac{\log 2}{\left[\log\left( 1 + Cx \right)\right]^2} \left[h(Cx) + \frac{C^2x^2}{(1+Cx)(2+Cx)}\right] \\
&\geq 0,
\end{align*}
since $h(Cx) \geq 0$ as shown in Appendix~\ref{appx:concave}. This complete the proof.

\balance
\bibliographystyle{IEEEtran}
\bibliography{navid}

\begin{thebibliography}{10}
\providecommand{\url}[1]{#1}
\csname url@samestyle\endcsname
\providecommand{\newblock}{\relax}
\providecommand{\bibinfo}[2]{#2}
\providecommand{\BIBentrySTDinterwordspacing}{\spaceskip=0pt\relax}
\providecommand{\BIBentryALTinterwordstretchfactor}{4}
\providecommand{\BIBentryALTinterwordspacing}{\spaceskip=\fontdimen2\font plus
\BIBentryALTinterwordstretchfactor\fontdimen3\font minus
  \fontdimen4\font\relax}
\providecommand{\BIBforeignlanguage}[2]{{%
\expandafter\ifx\csname l@#1\endcsname\relax
\typeout{** WARNING: IEEEtran.bst: No hyphenation pattern has been}%
\typeout{** loaded for the language `#1'. Using the pattern for}%
\typeout{** the default language instead.}%
\else
\language=\csname l@#1\endcsname
\fi
#2}}
\providecommand{\BIBdecl}{\relax}
\BIBdecl

\bibitem{Goodfellow-et-al-2016}
I.~Goodfellow, Y.~Bengio, and A.~Courville, \emph{Deep Learning}.\hskip 1em
  plus 0.5em minus 0.4em\relax MIT Press, 2016,
  \url{http://www.deeplearningbook.org}.

\bibitem{lecun2015deep}
Y.~LeCun, Y.~Bengio, and G.~Hinton, ``Deep learning,'' \emph{Nature}, vol. 521,
  no. 7553, p. 436, 2015.

\bibitem{krizhevsky2012imagenet}
A.~Krizhevsky, I.~Sutskever, and G.~E. Hinton, ``Imagenet classification with
  deep convolutional neural networks,'' in \emph{Advances in neural information
  processing systems}, 2012, pp. 1097--1105.

\bibitem{sutskever2014sequence}
I.~Sutskever, O.~Vinyals, and Q.~V. Le, ``Sequence to sequence learning with
  neural networks,'' in \emph{Proceedings of the 27th International Conference
  on Neural Information Processing Systems - Volume 2}, ser. NIPS’14.\hskip
  1em plus 0.5em minus 0.4em\relax Cambridge, MA, USA: MIT Press, 2014, p.
  3104–3112.

\bibitem{leung2014deep}
M.~K. Leung, H.~Y. Xiong, L.~J. Lee, and B.~J. Frey, ``Deep learning of the
  tissue-regulated splicing code,'' \emph{Bioinformatics}, vol.~30, no.~12, pp.
  i121--i129, 2014.

\bibitem{soups}
\BIBentryALTinterwordspacing
P.~E. Naeini, S.~Bhagavatula, H.~Habib, M.~Degeling, L.~Bauer, L.~F. Cranor,
  and N.~Sadeh, ``Privacy expectations and preferences in an {IoT} world,'' in
  \emph{Thirteenth Symposium on Usable Privacy and Security ({SOUPS}
  2017)}.\hskip 1em plus 0.5em minus 0.4em\relax Santa Clara, CA: {USENIX}
  Association, Jul. 2017, pp. 399--412. [Online]. Available:
  \url{https://www.usenix.org/conference/soups2017/technical-sessions/presentation/naeini}
\BIBentrySTDinterwordspacing

\bibitem{poushter2016smartphone}
J.~Poushter \emph{et~al.}, ``Smartphone ownership and internet usage continues
  to climb in emerging economies,'' \emph{Pew Research Center}, vol.~22, pp.
  1--44, 2016.

\bibitem{mcmahan2016communication}
H.~B. McMahan, E.~Moore, D.~Ramage, S.~Hampson \emph{et~al.},
  ``Communication-efficient learning of deep networks from decentralized
  data,'' \emph{arXiv preprint arXiv:1602.05629}, 2016.

\bibitem{kamp2018efficient}
M.~Kamp, L.~Adilova, J.~Sicking, F.~H{\"u}ger, P.~Schlicht, T.~Wirtz, and
  S.~Wrobel, ``Efficient decentralized deep learning by dynamic model
  averaging,'' in \emph{Joint European Conference on Machine Learning and
  Knowledge Discovery in Databases}.\hskip 1em plus 0.5em minus 0.4em\relax
  Springer, 2018, pp. 393--409.

\bibitem{zhang2013communication}
Y.~Zhang, J.~C. Duchi, and M.~J. Wainwright, ``Communication-efficient
  algorithms for statistical optimization,'' \emph{The Journal of Machine
  Learning Research}, vol.~14, no.~1, pp. 3321--3363, 2013.

\bibitem{chen2018lag}
T.~Chen, G.~Giannakis, T.~Sun, and W.~Yin, ``{LAG}: Lazily aggregated gradient
  for communication-efficient distributed learning,'' in \emph{Advances in
  Neural Information Processing Systems}, 2018, pp. 5050--5060.

\bibitem{scaman2018optimal}
K.~Scaman, F.~Bach, S.~Bubeck, L.~Massouli{\'e}, and Y.~T. Lee, ``Optimal
  algorithms for non-smooth distributed optimization in networks,'' in
  \emph{Advances in Neural Information Processing Systems}, 2018, pp.
  2740--2749.

\bibitem{koloskova2019decentralized}
A.~Koloskova, S.~U. Stich, and M.~Jaggi, ``Decentralized stochastic
  optimization and gossip algorithms with compressed communication,''
  \emph{arXiv preprint arXiv:1902.00340}, 2019.

\bibitem{wang2019matcha}
J.~Wang, A.~K. Sahu, Z.~Yang, G.~Joshi, and S.~Kar, ``{MATCHA}: Speeding up
  decentralized {SGD} via matching decomposition sampling,'' \emph{arXiv
  preprint arXiv:1905.09435}, 2019.

\bibitem{basu2019qsparse}
D.~Basu, D.~Data, C.~Karakus, and S.~Diggavi, ``Qsparse-local-{SGD}:
  Distributed {SGD} with quantization, sparsification, and local
  computations,'' \emph{arXiv preprint arXiv:1906.02367}, 2019.

\bibitem{reisizadeh2019robust}
A.~Reisizadeh, H.~Taheri, A.~Mokhtari, H.~Hassani, and R.~Pedarsani, ``Robust
  and communication-efficient collaborative learning,'' in \emph{Advances in
  Neural Information Processing Systems}, 2019, pp. 8386--8397.

\bibitem{amiri2019machine}
M.~M. Amiri and D.~Gunduz, ``Machine learning at the wireless edge: Distributed
  stochastic gradient descent over-the-air,'' \emph{arXiv preprint
  arXiv:1901.00844}, 2019.

\bibitem{ahn2019wireless}
J.-H. Ahn, O.~Simeone, and J.~Kang, ``Wireless federated distillation for
  distributed edge learning with heterogeneous data,'' in \emph{2019 IEEE 30th
  Annual International Symposium on Personal, Indoor and Mobile Radio
  Communications (PIMRC)}.\hskip 1em plus 0.5em minus 0.4em\relax IEEE, 2019,
  pp. 1--6.

\bibitem{zeng2019energy}
Q.~Zeng, Y.~Du, K.~K. Leung, and K.~Huang, ``Energy-efficient radio resource
  allocation for federated edge learning,'' \emph{arXiv preprint
  arXiv:1907.06040}, 2019.

\bibitem{yang2019scheduling}
H.~H. Yang, Z.~Liu, T.~Q. Quek, and H.~V. Poor, ``Scheduling policies for
  federated learning in wireless networks,'' \emph{IEEE Transactions on
  Communications}, 2019.

\bibitem{amiri2020update}
M.~M. Amiri, D.~Gunduz, S.~R. Kulkarni, and H.~V. Poor, ``Update aware device
  scheduling for federated learning at the wireless edge,'' \emph{arXiv
  preprint arXiv:2001.10402}, 2020.

\bibitem{geng2015optimality}
C.~Geng, N.~Naderializadeh, A.~S. Avestimehr, and S.~A. Jafar, ``On the
  optimality of treating interference as noise,'' \emph{IEEE Transactions on
  Information Theory}, vol.~61, no.~4, pp. 1753--1767, 2015.

\bibitem{naderializadeh2014itlinq}
N.~Naderializadeh and A.~S. Avestimehr, ``{ITLinQ}: A new approach for spectrum
  sharing in device-to-device communication systems,'' \emph{IEEE Journal on
  Selected Areas in Communications}, vol.~32, no.~6, pp. 1139--1151, 2014.

\bibitem{gupta1999critical}
P.~Gupta and P.~R. Kumar, ``Critical power for asymptotic connectivity in
  wireless networks,'' in \emph{Stochastic analysis, control, optimization and
  applications}.\hskip 1em plus 0.5em minus 0.4em\relax Springer, 1999, pp.
  547--566.

\bibitem{decreusefond:hal-00864303}
\BIBentryALTinterwordspacing
L.~Decreusefond, P.~Martins, and A.~Vergne, ``{Clique number of random
  geometric graphs},'' 2013, working paper or preprint. [Online]. Available:
  \url{https://hal.archives-ouvertes.fr/hal-00864303}
\BIBentrySTDinterwordspacing

\bibitem{mcdiarmid2011chromatic}
C.~McDiarmid and T.~M{\"u}ller, ``On the chromatic number of random geometric
  graphs,'' \emph{Combinatorica}, vol.~31, no.~4, pp. 423--488, 2011.

\bibitem{muller2008two}
T.~M{\"u}ller, ``Two-point concentration in random geometric graphs,''
  \emph{Combinatorica}, vol.~28, no.~5, p. 529, 2008.

\end{thebibliography}

\end{document}